\documentclass[letterpaper, 10 pt, conference]{ieeeconf}  

\IEEEoverridecommandlockouts                              

\usepackage{graphicx}
\usepackage{amsmath, bm}
\usepackage{color}
\usepackage{float}
\usepackage{mathtools}
\usepackage[ruled,vlined,linesnumbered]{algorithm2e}

\usepackage{graphics} 
\usepackage{epsfig} 
\usepackage{times} 
\usepackage{amsmath} 
\usepackage{amssymb}  
\usepackage{color}
\usepackage[noend]{algpseudocode}
\usepackage{booktabs,lipsum}
\usepackage{accents}
\usepackage{cite}
\usepackage{amsthm}
\usepackage[dvipsnames]{xcolor}
\usepackage{hyperref}
\usepackage[capitalise]{cleveref}
\usepackage[normalem]{ulem}

\usepackage{dsfont}

\newcommand{\mc}[1]{\mathcal{#1}}

\newcommand{\graph}{\mc{G}}

\newtheorem{lemma}{Lemma}

\newtheorem{theorem}{Theorem}
\newtheorem{corollary}{Corollary}
\newtheorem{remark}{Remark}
\newtheorem{assumption}{Assumption}

\def\graph{\mc{G}}
\def\reals{\mathbb{R}}

\def\dstate{x}
\def\astate{y}
\def\visregion{\mathcal{U}_k}

\def\path{\mathcal{P}}
\def\safeset{\mathcal{C}}
\def\frontier{\mathcal{F}}
\DeclareMathOperator{\sign}{sign}

\newif\iflong
\longtrue 

\newcommand{\lemmaproofonline}{The proof can be found in the full-length version online at \url{https://arxiv.org/abs/2204.04176}.\gdef\lemmaproofonline{See full-length version online.}}

\makeatletter
\newcommand{\removelatexerror}{\let\@latex@error\@gobble}
\makeatother

\title{Path Defense in Dynamic Defender-Attacker Blotto Games (dDAB)\\ with Limited Information}

\author{Austin K. Chen$^{1*}$, Bryce L. Ferguson$^{2*}$, Daigo Shishika$^{3}$, Michael Dorothy$^{4}$,\\ Jason R. Marden$^{2}$, George J. Pappas$^{1}$, Vijay Kumar$^{1}$
\thanks{$^*$The first two authors contributed equally as co-first authors.}
\thanks{$^{1}$Austin, George, and Vijay are with the University of Pennsylvania
        {\tt\small \{akchen, pappasg, kumar\}@seas.upenn.edu }}%
\thanks{$^{2}$Bryce and Jason are with the University of California, Santa Barbara
        {\tt\small \{blferguson,  jrmarden\}@ece.ucsb.edu }}
\thanks{$^{3}$ Daigo is with George Mason University {\tt\small dshishik@gmu.edu }}
\thanks{\raggedright $^{4}$ Michael is with DEVCOM Army Research Laboratory 
{\tt\small michael.r.dorothy.civ@army.mil}}
\thanks{The views expressed in this paper are those of the authors and do not reflect the official policy or position of the United States Government, Department of Defense, or its components. We gratefully acknowledge the support from 
ARL Grant DCIST CRA W911NF-17-2-0181, 
NSF Grant CCR-2112665, ONR grants N00014-20-1-2822 
and N00014-20-S-B001, 
and Lockheed Martin. 
}
}

\begin{document}

\maketitle
\thispagestyle{empty}
\pagestyle{empty}

\begin{abstract}

We consider a path guarding problem in dynamic Defender-Attacker Blotto games (dDAB), where a team of robots must defend a path in a graph against adversarial agents.
Multi-robot systems are particularly well suited to this application, as recent work has shown the effectiveness of these systems in related areas such as perimeter defense and surveillance.
When designing a defender policy that guarantees the defense of a path, information about the adversary and the environment can be helpful and may reduce the number of resources required by the defender to achieve a sufficient level of security.
In this work, we characterize the necessary and sufficient number of assets needed to guarantee the defense of a shortest path between two nodes in dDAB games when the defender can only detect assets within $k$-hops of a shortest path.
By characterizing the relationship between sensing horizon and required resources, we show that increasing the sensing capability of the defender greatly reduces the number of defender assets needed to defend the path.
\end{abstract}






\section{Introduction}
The emergence of new technologies in multi-robot systems and their applications in surveillance \cite{renzaglia2012multi}, pick-and-place~\cite{bozma2012multirobot}, delivery \cite{mathew2015planning}, etc., have motivated a large area of research on determining how to delegate tasks to each robot \cite{korsah2013comprehensive, khamis2015multi} and how to allocate robotic resources to different regions~\cite{Ferguson2021robust, Marden2017}.
With an increased understanding of how to perform this task assignment comes improvements in the overall operation of the system and the potential to complete jobs with less physical resources.


One particular area where multi-robot systems can offer new opportunities is in environments with adversarial operators~\cite{agmon2008multi}.
The use of security or defense systems has been studied in many different contexts including cross-fire attacks in network-routing~\cite{Kang2013,grimsman2020stackelberg}, security against malicious groups~\cite{brown2006defending}, defending networks of sub-systems from multiple attackers~\cite{Abdallah2021,kovenock2018optimal}, and many more.
Within each of these settings, the defender's ability to offer security guarantees in the face of unknown adversarial actions is hard~\cite{Duvallet2010} and depends on their knowledge of the system environment and the adversary's capabilities.
In this work, we seek to understand how increasing information about the adversary and the environment can improve a defender's ability to provide security guarantees with limited resources.

\begin{figure}
    \centering
    \includegraphics[width=0.95\linewidth]{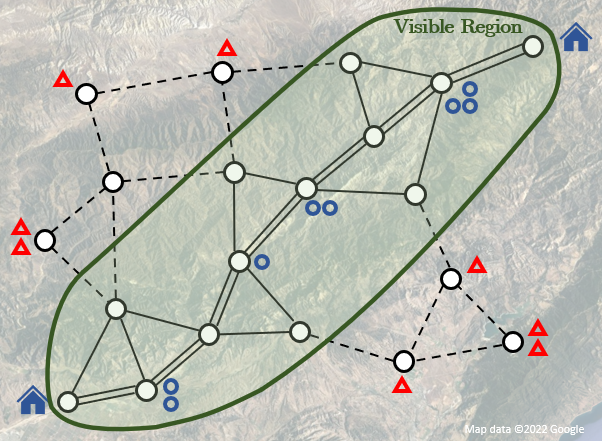}
    \caption{Illustration of the limited-visibility path defense problem. At each node, the defender (blue) and the attacker (red) each posses a number of assets that they sequentially maneuver through the network. The objective of the defender is to guarantee each node on the double-line path has more defender assets than attacker assets at every time step. The defender can only detect attacker assets within its sensing horizon, represented as the shaded green region.}
    \label{fig:game_illustration}
    \vspace{-5mm}
\end{figure}



The interactions between a defender and an attacker (or adversary) can be described by a two-player zero-sum game, where each of the decision-makers' objectives is inversely aligned.
One model that captures the key principles of these interactions are Colonel Blotto games~\cite{roberson2006colonel}, where two players each posses a finite reserve of troops that they allocate to various battlefields; at each battlefield, whichever player has allotted more troops wins the battlefield.
The Colonel Blotto game has been used to develop algorithms and deploy security strategies in many domains including airport security, border control, and wildlife protection~\cite{tambe2011security,pita2008deployed,xu2020stay}.
Additionally, researchers have used these games to study the interactions of defenders and attackers when there is incomplete information about the value of battlefields~\cite{adamo2009blotto,kovenock2011blotto} or the budget of the opposing colonel~\cite{paarporn2021general,Paarporn2019b}.
Though these results provide a first glimpse at how information affects defender decision making, the results focus on each colonel's ability to win battlefields in a one-shot setting. 
Instead, we wish to investigate the conditions under which a defender may guarantee a defense objective against a dynamic adversary.


In this work, we study the dynamic Defender-Attacker Blotto (dDAB) game where a defender and attacker sequentially maneuver assets in a network; on each node, whichever decision maker possesses more assets takes control of it.
Originally introduced in \cite{shishika2021dynamic}, the dDAB game was used to address what strategies a team of robots can use to defend every node in a network from adversarial intervention.
Here, we focus on the path guarding problem, where the defender must maintain control of each node in a shortest path between an origin and destination (\cref{fig:game_illustration}).
We then address how increased information about the adversary and the surrounding environment (in terms of visibility from the shortest path) affects the amount of defender resources required to obtain security guarantees.
This work is closely connected to perimeter defense problems~\cite{shishika2020review, chen2014path, garcia2019strategies, ESLee2020} and especially their multi-defender versions \cite{Macharet2020, AKChen2021,Shishika2021p}; however, here we specifically focus on paths in networks and introduce a setting where the defender has limited information about the adversary.

For a defender with a limited sensing horizon (measured by how many hops from the chosen path the defender can detect an adversary, potentially using implemented surveillance equipment), we characterize the necessary and sufficient number of defender assets required to guarantee that the defender can maintain control of a shortest path between a start and a target node\iflong \ (\cref{fig:big_net})\fi.
We also provide an initial deployment and algorithm that realizes this guarantee.
Our result shows that as the sensing horizon (and defender information) increases, the required number of assets to guarantee defense decreases.

\section{Problem Formulation}


In a finite graph $\graph= (\mathcal{V}, \mathcal{E})$, we consider the problem of path defense, where a defender seeks to defend a predetermined path $\path$ from an adversary.
We define the path $\path$ with cardinality $|\path |$ as a path graph with vertices $p_1, \dots, p_{|\path|}$ and $|\path|-1$ edges denoted by $(p_1, p_2), (p_2, p_3) \dots, (p_{|\path|-1}, p_{|\path|})$. We assume $|\path|\geq 3$ to avoid trivially short paths. Denote the start node $S$ and the target $T$ as the first and last nodes in the path respectively, so that $p_1 \equiv S$ and $p_{|\path|} \equiv T$. We will now make an assumption on the structure of $\path$, using the distance $d(v_1, v_2)$ to denote the minimum length of any path (as measured by the number of edges) between $v_1$ and $v_2$:

\begin{assumption} \label{as:shortest_path}
$\path$ is a shortest $S-T$ path in $\graph$, i.e.
\begin{equation}\label{eq:shortest_path}
     d(S,T) = |\path| - 1.
\end{equation}
\end{assumption}
%
Together, $\path$ and $\graph$ define the environment in which a specific instance of the game is played. Since most of the operations and functions that follow depend on these parameters, we will omit the dependence on $\path$ and $\graph$ when the relationship is clear.

The path defense game is played by a defender and an adversary. Both players have a finite amount of resources (or assets) at their disposal, denoted by $X \in \reals_{> 0}$ for the defender and $Y \in \reals_{> 0}$ for the adversary. Each player distributes their assets over the vertices of $\graph$ in a deterministic and centralized fashion. The resulting asset distribution for each player is a point within the standard simplex of dimension $|\mathcal{V}|$, scaled so that the sum over elements is equal to the total resources of the player. For the defender, the asset distribution $x$ is a vector defined as
\begin{equation}
    \dstate \in \Big\{ z \in \mathbb{R}^{|\mathcal{V}|}_{\geq 0}\ \Big| \ \sum_{i=1}^{|\mathcal{V}|} z_i = X, z_i \geq 0\Big\}.
\end{equation}
The adversary asset distribution $y$ is defined as above with $X$ replaced by $Y$. The scalar amounts of resources for the defender and adversary at vertex $v$ are defined as $x_v$ and $y_v$ respectively.

The game terminates with an adversary win if the adversary has more assets than the defender on any of the path nodes, i.e. if $\exists v \in \path$ such that $y_v > x_v$. Otherwise, the path is defended. 
We define the \emph{safe set}, $\safeset$, as follows:
\begin{equation}
    \safeset = \left\{[x,y]\;|\;x_v \geq y_v,\forall v\in\path \right\}.
\end{equation}

\iflong
\begin{figure}
    \centering
    \includegraphics[width=0.48\textwidth]{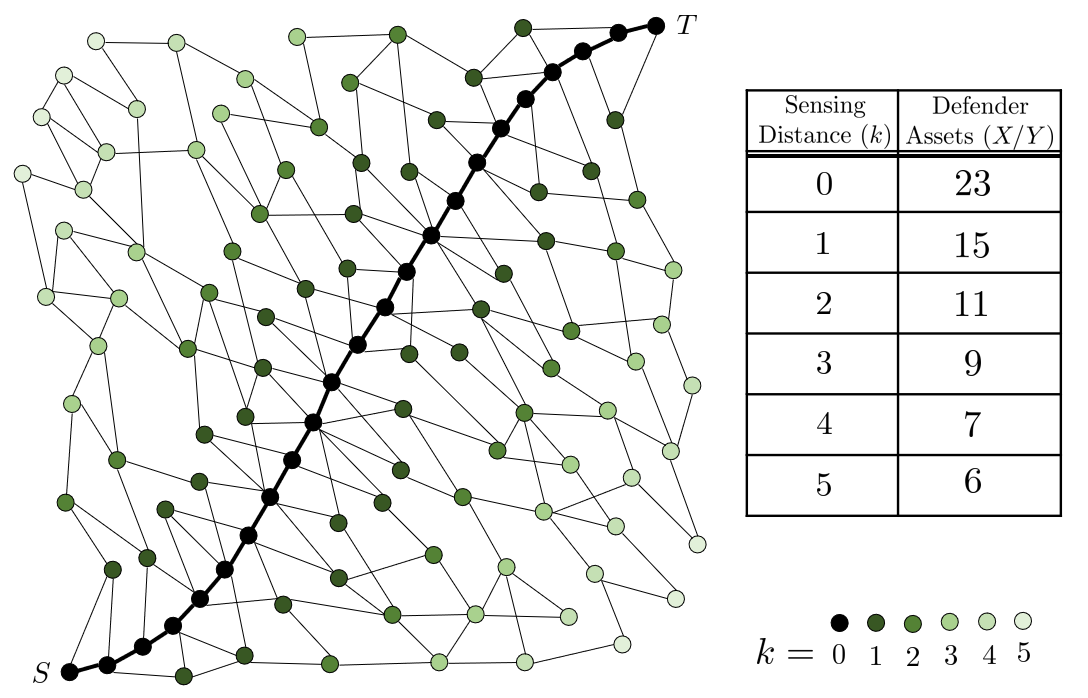}
    \caption{Visibility in a larger network, where a node's distance from the defended, shortest path is denoted by its color. As the sensing distance $k$ increases, the defender requires fewer assets to defend the path. The necessary and sufficient number of defender assets ($X$), relative to the number of attacker assets ($Y$) are given in the table above and are precisely those defined in \eqref{eq:assets_bound} where here $|\mathcal{P}|=23$.}
    \label{fig:big_net}
    \vspace{-5mm}
\end{figure}
\fi

In this paper, we consider the dynamic version of the path defense game, where the game is played out over a series of timesteps $t$. Accordingly, the defender and adversary allocations become time-varying vectors $x(t)$ and $y(t)$ respectively. When the game begins, the defender and adversary select some initial states $x(0)$ and $y(0)$. We wish to consider only non-trivial starting conditions, i.e. we assume that $[x(0), y(0)] \in \safeset$. Because the adversary can choose any arbitrary $y(0)$, we will allow the defender to observe $y(0)$ before deciding its own initial state $x(0)$.

Given states $x(t)$ and $y(t)$, the game timestep $t$ consists of the defender first transitioning its assets according to the function
\begin{equation}
    \dstate(t+1) = K^D(t) \cdot \dstate(t),
\end{equation}
where $K^D \in \reals^{|\mathcal{V}| \times |\mathcal{V}|}$ is a column stochastic state transition matrix (this enforces that the total number of defender assets remains unchanged over time). All elements of $K^D$ must be nonnegative and entry $K^D_{i,j} > 0$ only if $(v_i, v_j) \in \mathcal{E}$.
These constraints capture the notion that at every time step the defender can only move its assets up to one hop away from their current positions, and that assets may only be transferred along edges in $\graph$.

\begin{figure}[t]
    \vspace{1.5mm}
    \centering
    \includegraphics[width=.8\linewidth]{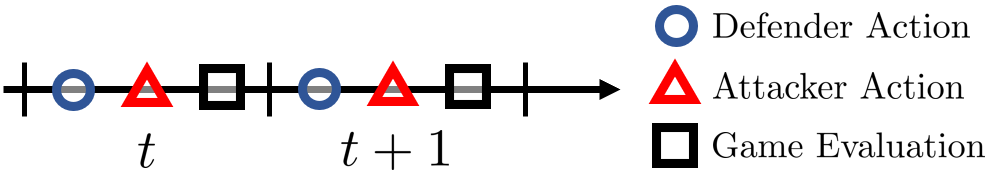}
    \caption{Sequence of events at every step of the game. The defender first moves its assets based on the current adversary state, after which the adversary observes before making its own move. Finally, the game outcome is evaluated.}
    \label{fig:game_sequence}
    \vspace{-5mm}
\end{figure}

In the same timestep $t$, after the defender transitions its assets, the adversary transitions its own assets according to 
\begin{equation}
    \astate(t+1) = K^A(t) \cdot \astate(t)
\end{equation}
and the result of the game is evaluated. If $[x(t+1), y(t+1)] \notin \safeset$, then the game is terminated and declared an adversary win. Otherwise, the game continues to the next timestep $t+1$. This order of play is shown in Figure~\ref{fig:game_sequence}.

We make no assumptions on how $K^A(t)$ is generated, but it must obey the graph and resource preservation constraints. Note that in this formulation, the problem takes the form of a Stackelberg game at every timestep since the adversary may observe the defender's action before taking its own. In contrast, the defender must generate its strategy without knowing where the attacker will move.

In this work, we are particularly interested in understanding the effect that information has on the defender's capability to maintain defense of the path.
While both players know the state of their own assets and the adversary can observe the state of the defender's assets, the defender may not be able to fully observe adversarial assets.
Let $k$ be the defender's sensing distance and define the visible region $\visregion \subseteq \mathcal{V}$ as the set of vertices where the defender is able to observe the adversary's assets. If we define the minimum path distance function $d^*(v)$ for a vertex $v\in\mathcal{V}$ as
\begin{equation}\label{eq:pathDistance}
    d^*(v) = \min_{p_i \in \path} d(v, p_i)
\end{equation}
then the visible region $\visregion$ for a specific $k$ is given as
\begin{equation}\label{eq:visibilityRegion}
    \visregion = \{v \in \mathcal{V} \ | \ d^*(v) \leq k\}.
\end{equation}
Since the defender can only observe adversary assets when they are in $\visregion$, the defender only observes a sub-vector of the adversary state $y(t)$. We call this sub-vector the observable adversary state $\widehat{y}(t)$ and define it as
\begin{equation}
    \widehat{y}(t) = \{ y_v(t) \ | \ v \in \visregion\}
\end{equation}
where $y_v(t)$ is the adversary assets at vertex $v$ during time $t$. We can now specify the form of $K^D(t)$ as
\begin{equation}
    K^D(t) = \pi^D(x(t), \widehat{y}(t))
\end{equation}
where $\pi^D \colon \reals^{2|\mathcal{V}|} \to \reals^{|\mathcal{V}| \times |\mathcal{V}| }$ is the control policy of the defender.

In this paper, we are interested in investigating how to determine $\pi^D$ and the minimum number of defender assets $X$ required to ensure that the path is always defended, i.e. that $[x(t), y(t)] \in \safeset \ \forall \ t$. We will also investigate how the required amount of defender assets changes as a function of the sensing distance $k$.

\subsection{Properties of Path Guarding Games}
We make a few observations about the dDAB problem as specified by the preceding problem formulation. 
First, no vertex $v\in\mathcal{V}\setminus \path$ can be connected to a pair of path nodes $p_i,p_j\in\path$ such that $d(p_i,p_j)>2$.  
This follows because of the assumption that $\path$ is a shortest path between $S$ and $T$ in $\graph$, as if it were not true then a shorter path would exist through $v$. Additionally, from this property, we may also see that each node may be connected to at most $3$ distinct path vertices, because otherwise the distance between the connected path vertices would exceed $2$. In this case, the $3$ path vertices must also be consecutive.

Another observation we make is that each vertex $v \in \mathcal{V}$ has its own static path distance $d^*(v)$. We may then organize the vertices in terms of their path distances. This also means that every visible region with a certain sensing distance is always a subset of a visible region with a greater sensing distance, i.e.
$
    k \leq m\implies \visregion \subseteq \mathcal{U}_m
$.
Intuitively, this means that as the sensing distance increases, the corresponding visible region grows to include nodes that are farther from the path, e.g. for $k=0$ the visible region is just $\path$, for $k=1$ the visible region is $\path$ and all 1-hop neighbors of $\path$, etc.

\begin{figure}
    \vspace{1.5mm}
    \centering
    \includegraphics[width=0.48\textwidth]{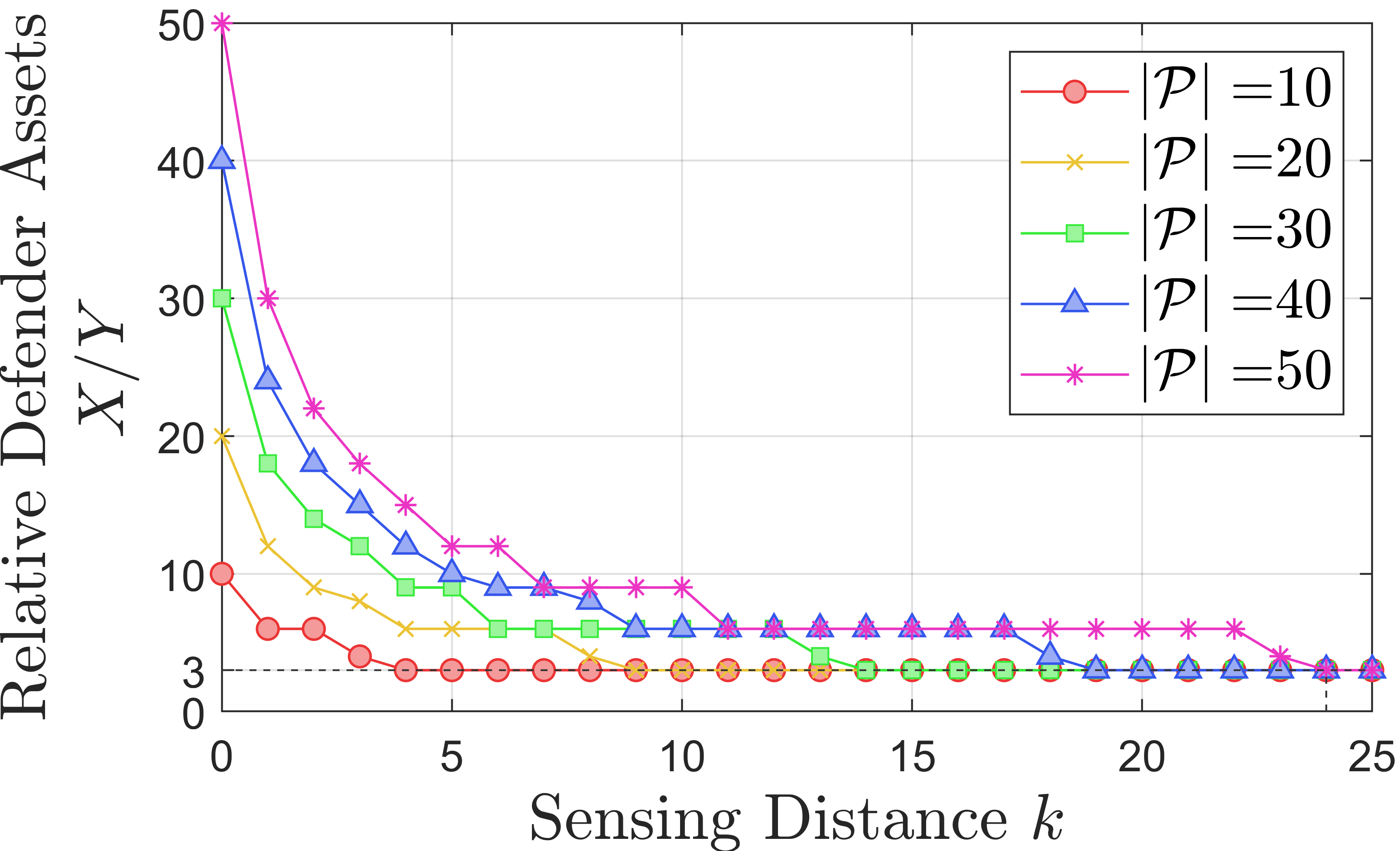}
    \caption{Number of defending assets needed over various sensing distances. The number of defender assets needed to guarantee defense on a shortest path $\path$ decreases as the sensing horizon increases. When the sensing horizon is $k>|\mathcal{P}|/2$ only $3Y$ defender assets are needed.}
    \label{fig:assets_needed}
    \vspace{-5mm}
\end{figure}

\section{Main Result}
We seek to understand how additional information can aid the defender when operating in an uncertain environment against an adversary.
In such settings, the environment (i.e. network structure) may influence how difficult it is for the defender to maintain their objective; however, it may be difficult or impossible for the defender to know this a priori.
We therefore derive a scheme for defense policies that will succeed in any, considered network structure.


To understand how information can help a defender maintain their objective of securing $\path$, \cref{thm:path_guard_vis} characterizes the necessary and sufficient amount of defender assets needed to guarantee that $\path$ is defended against any adversarial strategies.


\begin{theorem}\label{thm:path_guard_vis}
Let $\path$ be a path to be defended which satisfies the shortest path condition \eqref{eq:shortest_path} in $\mathcal{G}$.
For a given defender sensing horizon $k$, the condition
\begin{equation}\label{eq:assets_bound}
\resizebox{.87\hsize}{!}{$
    X \geq \left( 3 \left\lfloor {\frac{|\mathcal{P}|}{2k+3}} \right\rfloor + \min \Big\{ \hspace{-1mm}\bmod \hspace{-1mm} (|\mathcal{P}|, 2k+3), 3 \Big\}\right) Y
    $}
\end{equation}
is sufficient for $\path$ to be defended in any graph $\mathcal{G}$ and necessary for $\path$ to be defended across all graph structures $\mathcal{G}$.
\end{theorem}
\noindent Here, $\bmod (\diamond,\blacksquare)$ denotes the remainder of the Euclidean division of $\diamond$ by $\blacksquare$. The proof of \cref{thm:path_guard_vis} appears in \cref{sec:proof} with various parts separated into subsections.

In \cref{fig:assets_needed}, we show how the number of assets needed, quantified by \eqref{eq:assets_bound}, decreases as the sensing horizon increases.
Quantifying this improvement provides insights into the returns of investing in additional information.

We can further our understanding of the value of information by comparing \eqref{eq:assets_bound} with the number of defender assets needed when there is no visibility constraint.

First, note that if the defender has full visibility (i.e., $k \rightarrow \infty$), then in any network $\mathcal{G}$ we have that $X \geq 3Y$ defender assets are necessary and sufficient to guarantee path guarding.
This follows when the sensing distance is great enough,
as for any sensing distance $k$ such that $2k+3 > |\path|$, then $X =3Y$ is necessary and sufficient from \eqref{eq:assets_bound}.

We can see that the defender need only have sensing distance $k = |\mathcal{P}|/2$ to guarantee defense with the same number of assets as the full visibility case.
Accordingly, in \cref{fig:assets_needed} we see that number of needed defender assets quickly decreases and saturates for $k \geq |\mathcal{P}|/2$.

\section{Proof of Theorem~\ref{thm:path_guard_vis}}\label{sec:proof}

To prove \cref{thm:path_guard_vis}, we start by assuming the adversary possesses a single, unsplittable asset; as such $Y=1$ and $X$ only takes integer values. 
We will show in \cref{subsec:gen} that the results from this approach can be generalized to the case where any number of adversary assets can be split into multiple subgroups, including fractional assets.

The proof proceeds in two major parts. In subsection~\ref{subsec:suff}, we construct a strategy which shows that~\eqref{eq:assets_bound} is sufficient for guarding, and thus, an upper bound on the number of defenders required. In subsection~\ref{subsec:nec}, we show that~\eqref{eq:assets_bound} is necessary for guarding, and thus, a lower bound on the number of defenders required. Together, these results quantify the relationship between the sensing distance and the amount of resources required to defend $\path$.
The proof of \cref{thm:path_guard_vis} relies on several lemmas, the proofs of these lemmas can be found in an online appendix.

\begin{figure}[t]
    \vspace{2mm}
    \centering
    \includegraphics[width=.49\textwidth]{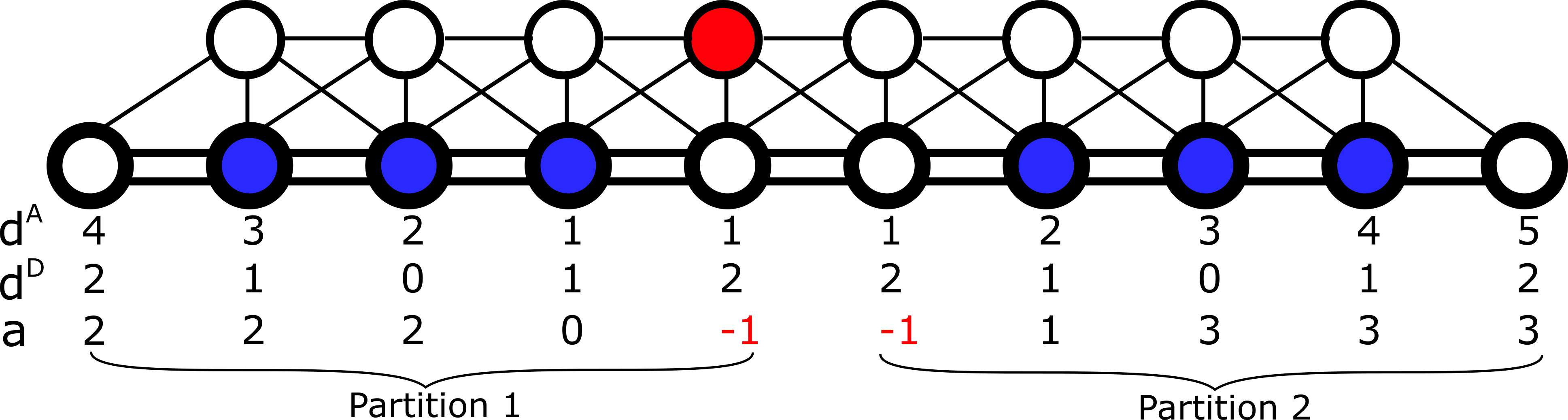}
    \caption{Example illustrating the distance $d^A$ from a single adversary, distance $d^D$ from the partition's center defender, and calculated advantage $a$ for each path node. Here, path nodes are shown by the bold circles at the bottom and are connected by double lines. Nodes with defender units are filled with blue, and the adversary's node is filled with red. Note that within each partition, defenders move as a platoon. In this example, $k=1$ and so the partition size is $2\cdot1 + 3 = 5$. Since there are negative advantage values (highlighted in red), the defenders in each partition must move towards the center so as to prevent the adversary from winning the game.}
    \label{fig:advantage_example}
    \vspace{-5mm}
\end{figure}

\subsection{Sufficient Algorithm}\label{subsec:suff}
In order to show that \eqref{eq:assets_bound} is sufficient, we present a defender algorithm which ensures that all nodes are defended. First, we split $\path$ into disjoint partitions of size $2k+3$ each. There may be one partition smaller than $2k+3$ if $\bmod(|\path|, 2k+3) \neq 0$. For partition $\omega$, we denote the center index as $c_\omega$, i.e. $p_{c_\omega}$ is the center node of the partition.  We call a \emph{platoon} a group of three unit defenders in consecutive nodes within a  partition, which always move together and maintain single spacing. The node index for the center of the platoon is termed $l_{\omega}$, and therefore one unit of defenders will be distributed at each of $p_{l_\omega - 1}, p_{l_\omega}, $ and $p_{l_\omega + 1}$. An example of this distribution is shown in Figure~\ref{fig:advantage_example}. For the special case where the partition is of size 1 or 2, there is one unit defender asset placed on each node within the partition, and the assets do not move. As a result, the nodes within this smaller partition are always guarded and so we do not consider this partition in our analysis below.


Given a path node $p_i$ and the adversarial asset's node location $v_A$, define the minimum adversary distance $d^A_i$ as
\begin{equation}
    d^A_i = d(p_i, v_A).
\end{equation}
Since the position of the defender's asset may vary over time, $d^A_i$ may also change over time but we omit this dependence for notational brevity when considering a single timestep. If the adversary's asset is unobserverable then $d^A_i = \infty$. Similarly, define the minimum defender distance $d^D_i$ as the minimum distance of the platoon center to a path node $p_i$ within the platoon's partition, i.e.
\begin{equation}
    d^D_i = d(p_i, l_\omega)
\end{equation}
for $p_i$ in partition $\omega$.
Then, we define the \emph{advantage} $a_i$ at path node $p_i$ as
\begin{equation}
    a_i = d^A_i - d^D_i.
\end{equation}
Example values of $d^A_i, d^D_i,$ and $a_i$ can be found in Figure~\ref{fig:advantage_example}. 
\begin{lemma} \label{lem:atk_win_nec}
Assume that the defender uses platoons to defend each partition. It is necessary in any adversary winning configuration that $a_i \leq -2$ for some $i \in [1, |\path |]$.
\end{lemma}
\begin{proof}
For the adversary to win, it must place its asset on some path node $p_i \in \path$, and therefore, $d^A_i = 0$. 

We observe that $d^D_i \geq 2$, since otherwise $p_i$ would be occupied by one of the three defender assets in the partition's platoon. Therefore $a_i \leq -2$.
\end{proof}
\begin{corollary} \label{cor:sufficient_condition}
Since Lemma~\ref{lem:atk_win_nec} describes a necessary condition for the adversary to win, a sufficient condition for the defender to successfully defend indefinitely is
\begin{equation}
    a_i \geq -1 \ \forall i \in [1, |\path | ],\forall t.
\end{equation}
\end{corollary}

\begin{figure}
 \vspace{1mm}
 \removelatexerror
  \begin{algorithm}[H]
  \label{alg:defender_step}
  \caption{DefenderStep transition function.}
\KwIn{Adversary position $n_A$, current platoon location indices $L = \{l_1, \dots, l_{|L|}\}$}
\KwOut{Updated platoon locations $L$}
\SetKwFunction{DEFENDERSTEP}{DefenderStep}
    \SetKwProg{Fn}{Function}{:}{}
    \For{Each partition $\{p_i, \dots, p_j\}$ with index $\omega$}
    {
    Compute advantages $\{a_i, \dots, a_j\}$\\
    Compute advantage frontier $\frontier$\\
    \If{$a_i > 0 \ \forall a_i \in \frontier$}
    {
        $l_\omega \longleftarrow l_\omega + \sign(c_\omega - l_\omega)$
    }
    \ElseIf{$\exists a_k \in \{a_i, \dots, a_j\}$ s.t. $a_k < 0$}{
    $l_\omega \longleftarrow l_\omega + \sign(k - l_\omega)$\\
    $l_\omega \longleftarrow \max(\min(l_\omega, j-1), i+1)$
    }
    }
    \textbf{Return} $L$
  \end{algorithm}
\vspace{-5mm}
\end{figure}

\begin{figure}
 \removelatexerror
  \begin{algorithm}[H]
  \label{alg:initialization}
   \caption{Initialization procedure.}
\KwData{{Adversary position $n_A$, platoon location indices $L = \{l_1, \dots, l_{|L|}\}$}}
Observe $n_A$\\
Initialize $L$ such that $l_\omega = c_\omega \ \forall \omega$\\
\While{$\exists a_i < -1$}
{
    $L \longleftarrow$ \Call{DefenderStep}{$n_A, L$}\\
}
Initialize platoons according to $L$\\
  \end{algorithm}
\vspace{-5mm}
\end{figure}

\textbf{Defender Strategy:} The defenders in each partition calculate the advantage values independently for the nodes in their own partitions to decide how to transition at each timestep. If the advantage value for any node in its partition is negative, then platoon moves towards that node along $\path$. Otherwise, the platoon considers the advantage values of all nodes (inclusive) between its own center and the closest partition boundary, which we call the set of frontier advantages $\frontier$. If every advantage value in $\frontier$ is positive, the platoon moves towards the middle of the partition. If neither of these conditions holds, then the defenders remain at their current position. Note that platoons are restricted to move within their partition, i.e. if partition $\omega$ consists of nodes $p_i$ through $p_j$ with $i < j$ then $l_\omega \in [i+1, j-1]$. This procedure, repeated once every timestep, is described in Algorithm~\ref{alg:defender_step}.

\textbf{Defender Initialization:} To initialize the defender positions, the defender may simply repeat this procedure until all advantage values are greater than or equal to $-1$. Then, the resulting configuration is selected as the initial distribution. This approach is described in Algorithm~\ref{alg:initialization}.

With this strategy, the platoon only moves in one of two directions, either towards $S$ or $T$. Because of this, it is not clear what should be done when negative advantage values appear on both sides of a platoon (i.e. if there is a negative advantage between the platoon and $S$, and also between the platoon and $T$). We will now show that this situation will never arise under the assumption that $\path$ is a shortest path.


\begin{lemma} \label{lem:def_move_valid}
Assume that $(\path,\graph)$ satisfies~\eqref{eq:shortest_path}. Then, for each partition, the defense strategy specified by Algorithm~\ref{alg:defender_step} will have negative advantage values on at most one side of the platoon within that partition. If a negative advantage value exists on one side of the platoon, the advantage values on the opposite side must all be positive.
\end{lemma}
\begin{proof}
\iflong
Suppose, towards a contradiction, that there are negative advantage values on one side of the platoon and nonpositive advantage values on the other side. Without loss of generality, denote the path vertex with negative advantage $p_i$ and path vertex with nonpositive advantage $p_j$. Call the center vertex of the platoon $l_\omega$ and the vertex of the adversary $v_A$. Since $a_i < 0$ and $a_j \leq 0$, we can say that $d^A_i < d^D_i$ and $d^A_j \leq d^D_j$ directly from the definition of the advantage. Adding these two inequalities yields $d^A_i + d^A_j < d^D_i + d^D_j$, meaning that the length of the path from $p_i$ to $p_j$ through $v_A$ is shorter than the path through $l_\omega$. Since $l_\omega$ must be located on $\path$ between $p_i$ than $p_j$, this contradicts the assumption that $\path$ is a shortest path and therefore negative advantage values will never appear on both sides of the platoon. If negative advantage values do appear, the advantage values within the partition on the other side of the platoon must be positive by similar reasoning.
\else \lemmaproofonline \fi
\end{proof}
\begin{corollary} \label{cor:def_overshoot}
When a platoon moves one step towards a negative advantage value $a_i$, no negative advantage values will appear on the opposite side (from $a_i$) of the platoon within the partition.
\end{corollary}
From Lemma~\ref{lem:def_move_valid}, we know that when a platoon moves towards negative advantage values, the advantage values within the partition on the opposite side must all be positive. Since the defender can only change any advantage value by $1$ with a single platoon move, the advantage values in question must be nonnegative after the move.

\begin{lemma} \label{lem:restore_original_config}
If the adversary moves out of $\visregion$ at timestep $t$, the platoons move back to the center of their partitions at timestep $t+1$.
\end{lemma}
\begin{proof}
\iflong

First, we will show that the distance of the platoon center $p_l$ from its partition center $p_c$ is bounded by a function of $d^*(v_A)$ where $v_A$ is the node location of the adversary. In particular, we claim that $d(p^+_l, p_c) \leq k - d^*(v_A) +1$, where $p^+_l$ is the center of the platoon after the defender moves its assets at any timestep. We see that this is true when $d^*(v_A) = k$, as the only vertices with possibly negative advantage values lie on the end of the partition. For all other vertices, $d^D_i \leq k$ and so $a_i \geq 0$. Therefore, the platoon will move at most 1 step away from the center to decrease $d^D_i$ from $k+1$ to $k$ for the end vertices, and we have $d(p^+_l, p_c) \leq k - k +1 = 1$.

To show that this is true for all $d^*(v_A)$, we show that increasing or decreasing the value of $d^*(v_A)$ preserves the bound. First, suppose that $d^*(v_A)$ decreases from $r$ to $r-1$. In the worst case, the platoon will be located at the edge of the bound such that $d(p_l, p_c) = k-r+1$, since all other initial locations will obey $d(p^+_l, p_c) \leq k-r+1$ after any 1-hop move. Therefore, at worst the platoon will move 1 more step away from $p_c$, resulting in a distance from the center of $d(p_l, p_c) = k-r + 1 + 1 = k - (r-1) + 1$, showing that the bound is preserved.

Now suppose $d^*(v_A)$ increases from $r$ to $r+1$. Consider the platoon at the edge of the proposed bound such that $d(p_l, p_c) = k-r+1$. We know that since the partition has $k+1$ nodes to either side of $p_c$, the distance to the closer end of the partition is bounded by $r$. Recall from \eqref{eq:visibilityRegion} that $d^*(v) \in [0, k] \ \forall v \in \visregion$ and  $d^*(v) \leq d(p_i, v) \ \forall p_i \in \path$. Therefore, when $d^*(v_A)$ increases from $r$ to $r+1$, $d^A_i \geq r+1$ for all path nodes and therefore the frontier advantage values are all positive, i.e. $a_i > 0 \ \forall \ a_i \in \frontier$. This means that under the proposed control law, the platoon will move towards $p_c$. If we instead consider a platoon located 1 step away from the boundary (i.e. $d(p_l, p_c) = k-r$) we see that the frontier advantages would be nonnegative, and so the platoon will not move towards its closest partition boundary. In either case, $d(p^+_l, p_c) \leq k - r$. Since all other initial platoon positions will satisfy $d(p^+_l, p_c) \leq k - r$ regardless of how they move, the bound holds when $d^*(v_A)$ decreases.

Since we know that the inequality $d(p^+_l, p_c) \leq k-d^*(v_A) + 1$ holds for the case when $d^*(v_A) = k$ and when $d^*(v_A)$ changes, we know it holds for all $d^*(v_A)$. Notice that before the adversary leaves the sensing radius, $d^*(v_A) = k$ and so $d(p^+_l, p_c) \leq 1$, implying that the platoon can return to $p_c$ within 1 step if the adversary leaves $\visregion$.
\else
\lemmaproofonline
\fi
\end{proof}

\begin{lemma} \label{lem:init_advantage}
Using Algorithm~\ref{alg:initialization}, the defender can achieve a starting configuration such that $a_i \geq -1 \ \forall \ i \in [1, |\path |]$ in finite time.
\end{lemma}
\begin{proof}
\iflong
Because of Lemma~\ref{lem:def_move_valid}, we know that each defender step is possible since there will never be negative advantage values on both sides of the platoon. As Corollary~\ref{cor:def_overshoot} states, we are also guaranteed that moving a platoon towards a negative advantage will never result in negative advantage values on the other side of the platoon after the move is made.

Additionally, the platoon will never have to leave its partition to achieve the stated condition since the platoon center can be moved to within one hop of any node in the partition. Therefore $d^D_i \leq 1$ can be achieved for any single partition node, bounding the advantage as $a_i \geq -1$. Thus, we conclude that repeatedly moving the platoons in the direction of negative advantage values will eventually result in a state where $a_i \geq -1 \ \forall \ i \in [1, |\path |]$. Since the number of nodes in the partition is finite, we also conclude that Algorithm~\ref{alg:initialization} terminates in finite time.
\else \lemmaproofonline \fi
\end{proof}

Since the proposed algorithm requires no more than the amount of resources specified by \eqref{eq:assets_bound}, we will now prove sufficiency for Theorem~\ref{thm:path_guard_vis} by showing that the proposed algorithm guarantees path guarding.
\begin{proof}[Proof of sufficiency for Theorem~\ref{thm:path_guard_vis}]
From Lemma~\ref{lem:init_advantage} we know that initially $a_i \geq -1 \ \forall \ i \in [1, |\path |]$.  During its turn, the defender can move towards the node with negative advantage value to restore the advantage values from $-1$ to $0$ if the adversary is not at the end of the partition. This is possible since the negative advantage value only appears on one side of the platoon as shown in Lemma~\ref{lem:def_move_valid}. Negative advantage values will also not appear on the opposite site of the platoon after the move, as stated in Corollary~\ref{cor:def_overshoot}. Any adversary move can change the advantage value at each node by at most $1$, and so the defender can repeat this procedure at every timestep with the same result.

For an adversary at the end of the partition, the platoon is not able to restore the advantage value for the edge of the partition to be nonnegative, as the platoon would have to move out of the partition to do so. However, in this case the partition is still guarded by one of the defender assets adjacent to the center of the platoon. The adversary is also not able to further decrease the advantage in this case, as either moving out of the partition or towards the platoon would increase the advantage value at the partition edge. 

Note that if the adversary asset leaves $\visregion$, it does not have to reappear at the same node from which it left. However, it is always true that $d*(v_A) = k$ whenever the asset first reappears after leaving $\visregion$, and from Lemma~\ref{lem:restore_original_config} we know that the $d^D_i \leq k+1$, so $a_i \geq -1$. Therefore, the defender can guarantee that $a_i \geq -1$ for all time. As stated in Corollary~\ref{cor:sufficient_condition}, this means that the defender can guard $\path$ indefinitely.
\end{proof}

\subsection{Necessary Defender Assets}\label{subsec:nec}
In this subsection, we show that \eqref{eq:assets_bound} is necessary to guarantee the defense of the path $\path$ when any additional graph structure can be realized outside of $\path$ such that no shorter path exists between $S$ and $T$.
In \cref{subsec:suff}, we provided a sufficient algorithm where we chose to place defender assets only on $\path$; we now show that any defense strategy that allocates assets throughout $\graph$ can be replicated by a defense strategy that allocates assets only on $\path$.
\begin{lemma}
Consider a defender policy $\pi^D$. The defender can guarantee the same defense of the path $\path$ in the graph $\mathcal{G}$ with the same number of assets $X$, by using a policy $\hat{\pi}^D$ that only has assets on $\path$ at each time step, i.e. $x_p(t) = 0$ if $p \notin \mathcal{P}$.
\end{lemma}

\begin{proof}
\iflong
Consider an attacker trajectory $\{y(t)\}_{t=0}^T$ and defender trajectory $\{x(t)\}_{t=1}^T$ that results from the defender policy $\pi^D$, where $T$ is the termination round\footnote{If the policy $\pi^D$ can guarantee defense in perpetuity, let $T\rightarrow \infty$.}.
Consider that the defender wins the nodes in the path $D_t \subseteq \mathcal{P}$ at each time $t \geq 1$.

Now, we define a new defender trajectory $\{\hat{x}\}_{t=1}^T$ that wins the same nodes $\{D_t\}_{t=1}^T$  and only places assets on $\path$, i.e., $\hat{x}_p(t)\geq x_p(t)$ if $p \in \mathcal{P}$ for all $t \in \{1,\ldots,T\}$, and $\hat{x}_p(t)=0$ if $p \notin \mathcal{P}$ for all $t \in \{1,\ldots,T\}$.
We can realize a policy that gives such a trajectory by using the algorithm from \cref{subsec:suff}.
For every original defender asset, generate an initial allocation and policy using the center asset position of Algorithm~\ref{alg:defender_step} while treating the old defender as the adversary.
If the new defender updates after the original defender action in the same timestep, then \cref{subsec:suff} shows that the new defender will coincide with the old defender on $\path$ (defending the same nodes $\{D_t\}_{t=1}^T$) while never leaving $\path$.
\else
\lemmaproofonline
\fi
\end{proof}

Next, we show a winning condition for the attacker based on the defender's starting configuration.
\begin{lemma}\label{lem:defender_three_config}
A winning attacker strategy exists in some graph $\mathcal{G}$ if there exist three consecutive nodes $p_{\alpha-1},p_\alpha,p_{\alpha+1}$ in the path $\path$ which the defender cannot move an asset onto each of in $k$ defender actions.
\end{lemma}

\begin{proof}
\iflong
Consider three consecutive nodes $\{p_{\alpha-1},p_{\alpha},p_{\alpha+1}\}$ in the path $\path$ where the defender cannot move one defender asset onto each node within $k$ time steps.
Construct $\graph$, such that there exists a node $p_{\xi} \in \mathcal{V} \backslash \path$ with edges $(p_{\xi},p_{\alpha-1}),(p_{\xi},p_{\alpha}),(p_{\xi},p_{\alpha+1}) \in \mathcal{E}$ and $p_{\xi}$ and can be reached from a node outside the sensing horizon in $k$-hops.
Node $p_\xi$ is one hop off $\path$, thus the attacker can move their asset from outside the region of visibility to $p_\xi$ in $k$ attacker actions, i.e., if the attacker starts this movement at time $t$, their asset will reach $p_\xi$ at time $t+k$.
The defender will first detect the attacker at time $t+1$.
If the defender could not move assets onto each of $\{p_{\alpha-1},p_{\alpha},p_{\alpha+1}\}$ in $k$ time steps, then at least one of these nodes will be uncovered by a defender asset in time $t+k+1$.
In the attacker's action at $t+k+1$, they can move their asset to any of $\{p_{\alpha-1},p_{\alpha},p_{\alpha+1}\}$ and thus could move to whichever node does not have a sufficient number of defender assets and take it.
\else
\lemmaproofonline
\fi
\end{proof}

Now, we prove that \eqref{eq:assets_bound} is necessary to guarantee defense across all graph structures.
Consider the case where the attacker starts outside of the region of visibility; because this is a valid strategy for the attacker, defending against it generates a necessary condition.
Following \cref{lem:defender_three_config}:
when $|\mathcal{P}| = 1$, one defending asset is necessary to defend $\path$,
when $|\mathcal{P}| = 2$, two defending assets are necessary to defend $\path$,
and when $|\mathcal{P}| = 3$, three defending assets are necessary to defend $\path$.
The following remark will allow us to compare the necessary number of defender assets in different paths.
\begin{remark}\label{rem:compare}
If $X$ defender assets are necessary to defend $\path$, then at least $X$ defender assets are necessary to defend $\path^\prime$ if $|\path^\prime|>|\path|$.
\end{remark}
Thus, at least three defending assets are necessary for all $|\mathcal{P}|>3$.
Next, consider two necessary conditions for defending subsets of $\path$:
\begin{enumerate}
\item A node $p \in \mathcal{P}$ is defended only if it is within $k$-hops of a defender asset.
\item For $i \in \{1,\ldots,|\path|-2\}$, the nodes $\{p_{i},p_{i+1},p_{i+2}\}$ are guaranteed-defended only if the nodes $\{p_{[i-k]_+},\ldots,p_{i+k+2}\}$ are initially allocated 3 defender assets.
\end{enumerate}
If the first were contradicted, an attacker asset could enter the region of visibility and reach node $p_i$ in $k$ attacker actions while the defender asset would require $k+1$ actions.
If the second were contradicted, then the defender would not be able to move assets into the three-in-a-row configuration described in \cref{lem:defender_three_config} over $\{p_{i},p_{i+1},p_{i+2}\}$ in $k$ defender turns and the attacker could take one of these nodes.

\iflong
\begin{figure*}[h!]
    \centering
    \includegraphics[width=.99\textwidth]{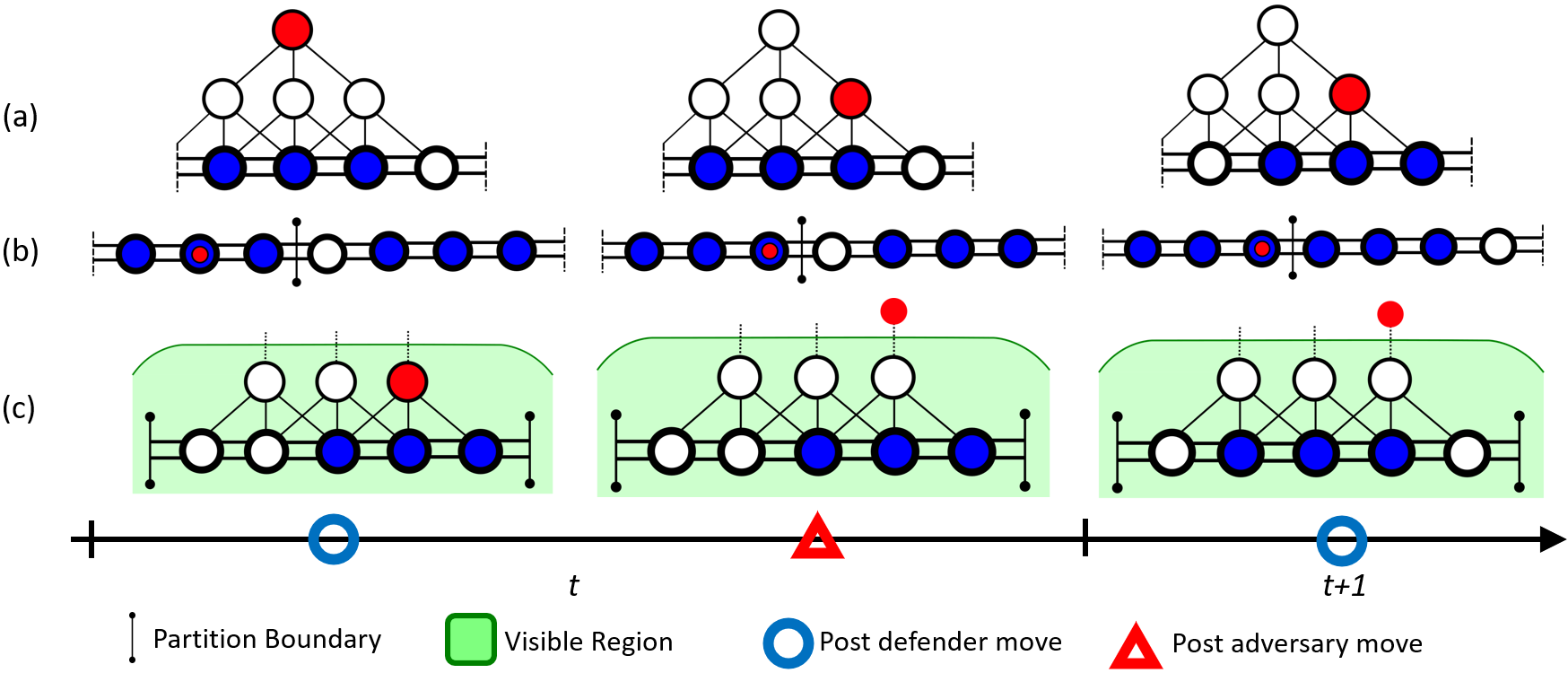}
    \caption{Illustrative scenarios under the proposed defender control algorithm. In (a), the defender platoon follows the adversary asset's movement to ensure guarding. In (b), two platoons move to the boundary between their partitions to ensure $\path$ is guarded when an adversary passes through a partition boundary. In (c), the platoon returns to the center of its partition when the adversary leaves $\visregion$, resetting its state so that it may guard against an asset reappearing anywhere within $\visregion$.}
    \label{fig:example_scenarios}
\end{figure*}
\fi

Now, for some $\beta \in \{1,\ldots,L\}$ consider the task of deploying additional defender assets when the defender's deployment on $\{p_1,\ldots,p_\beta\}$ is already given. 
Regardless of how many, or where, the assets in $\{p_1,\ldots,p_\beta\}$ are initially deployed, they cannot reach node $p_{\beta+k+1}$ in $k$-hops, so (from necessary condition 1), $p_{\beta+k+1}$ is not defended by defender assets in $\{p_1,\ldots,p_\beta\}$, and (from necessary condition 2) there must be three assets deployed on nodes $\{p_{\beta+1},\ldots,p_{\beta+2k+3}\}$, regardless of the deployment of assets on nodes $\{p_1,\ldots,p_{\beta}\}$.
As such, because we can pick any $\beta \in \{0,\ldots,L-1\}$ any continuous sequence of nodes of length $2k+3$ must have at least three assets deployed on them.
From this, a shortest path of length $|\mathcal{P}|=(2k+3)n$ must have $3n$ defender assets, as we can form $n$ disjoint, connected sequences of nodes in $\path$ that must each possess 3 defending assets, or $X \geq \left(3 \frac{L}{2k+3}\right)Y$.

Next, consider $|\mathcal{P}| = (2k+3)n + 1$ and $\{p_1,\ldots,p_{(2k+3)n}\}$ are fully defended by $3n$ assets.
From necessary condition 2, there must be 3 defending assets within the first $k+3$ nodes.
From necessary condition 1, these assets cannot reach (and thus defend) any node further than $p_{2k+3}$,
i.e., node $p_{2k+3+1}$ cannot be reached (or defended) by the first 3 defender assets.
If $n>1$, from necessary condition 2, nodes $\{p_{k+3+1},\ldots,p_{3k+6}\}$ must posses 3 defending assets, and from necessary condition 1, these assets cannot reach (and thus defend) any node further than $p_{2(2k+3)}$; so, node $p_{2(2k+3)+1}$ cannot be reached (or defended) by the first 6 defender assets.
Following this logic for any $n\geq 1$, we can see that $3n$ assets that fully defend the first $(2k+3)n$ nodes cannot reach $p_{(2k+3)n+1}$ and one additional asset is needed.

Similarly, if $|\mathcal{P}| = (2k+3)n + 2$ or $|\path| = (2k+3)n + 3$, and the $\{p_1,\ldots,p_{(2k+3)n}\}$ are fully defended by $3n$ assets, then nodes $p_{(2k+3)n + 1}$, $p_{(2k+3)n + 2}$, and $p_{(2k+3)n + 3}$ cannot be reached by the first $3n$ assets and 2 or 3 additional assets are needed respectively.

From \cref{rem:compare}, 
when $|\mathcal{P}| \in \{(2k+3)n + 3, \ldots, (2k+3)(n+1)\}$, the same number of defender assets ($3n+3$) are necessary. 
Therefore, the necessary number of assets is as expressed in \eqref{eq:assets_bound}.

\subsection{Generalization to Fractional Attacker Strategies}\label{subsec:gen}
We now relax our assumption that attacker has an unsplittable unit mass of assets to extend the results of \cref{subsec:suff} and \cref{subsec:nec} to the general adversary problem. Since there are no restrictions on the number of assets that may pass through an edge, and since the motions of assets are not affected by the motions of other assets, we may simply consider a series of parallel games for each adversary asset.
For example, if the adversary divides its assets in a 70:30 split between two nodes, the defender can respond by first splitting its assets (independently at each node) proportionally in a 70:30 split and then playing two independent games.
The first 70\% of defenders track the 70\% of adversaries, and the second group of 30\% defenders follow the 30\% of adversaries.
Additional splits by the adversary can prompt this defender response recursively.
Since the graph is finite, the number of games the defender must play is also finite.

Because the attacker strategy of moving each of their assets in the same mass remains viable, the necessary number of defender assets remains the same and matches the sufficient number proposed by \cref{subsec:suff}.

Finally, note that a split of adversarial assets in the unobservable region does not change this result. The defender naturally does not react to adversarial actions outside the unobservable region, and once the adversarial assets appear in the visible region path defense can be treated as before by allocating defender resources proportionally and playing multiple independent games.

\iflong
\section{Illustrative Scenarios}
We now show example scenarios of potential adversarial actions and the defender's response under the proposed sufficient defender transition function (Alg.~\ref{alg:defender_step}). Consider the states presented in Fig~\ref{fig:example_scenarios}, which are taken after the defender action at time $t$, after the adversary action at time $t$, and after the defender action at time $t+1$.

In scenario (a), the adversary is shown moving toward $\path$ to the right of the defender platoon. As expected, the defender also shifts its position towards the right at time $t+1$ to guard against the adversary asset.

In scenario (b), the adversary is already located on $\path$ but is guarded by the platoon in the left partition. As it moves towards the boundary between partitions, the platoon in the right partition responds by moving towards the partition boundary, ensuring that the path remains guarded. This example illustrates the hand-off of defending adversarial assets between platoons of different partitions.

In scenario (c), the adversary leaves the visible region, which is shaded in green. Note that at time $t$, the platoon is not located at the center of its partition. However, once the adversary leaves $\visregion$, the platoon restores its position to the center of the partition, thereby ensuring path defense should the adversary's assets reappear at another location.
\fi

\section{Conclusions and Future Work}
In this paper, we investigated the necessary and sufficient number of defenders required to defend a path on an arbitrary graph in the dDAB game as a function of the defender's sensing distance. 
We derived an expression which lower bounds the amount of defender resources required to defend a shortest path as a function of the sensing distance, and also described a defender strategy that guarantees guarding of the path while matching our lower bound.
Together, these results quantify the relationship between the sensing distance and the resources required to defend a shortest path.


Future work will examine key variations to the dDAB game such as additional motion models for the players (e.g., allowing the adversary to move faster than the defender) as well as different objective-region types (e.g., cycles). 


\iflong
\section*{Acknowledgement}
The authors would like to thank Lifeng Zhou, Yue Guan, and Kirby Overman for insightful discussions and valuable feedback.
\fi

\bibliographystyle{IEEEtran}
\bibliography{library}

\begin{thebibliography}{10}
\providecommand{\url}[1]{#1}
\csname url@samestyle\endcsname
\providecommand{\newblock}{\relax}
\providecommand{\bibinfo}[2]{#2}
\providecommand{\BIBentrySTDinterwordspacing}{\spaceskip=0pt\relax}
\providecommand{\BIBentryALTinterwordstretchfactor}{4}
\providecommand{\BIBentryALTinterwordspacing}{\spaceskip=\fontdimen2\font plus
\BIBentryALTinterwordstretchfactor\fontdimen3\font minus
  \fontdimen4\font\relax}
\providecommand{\BIBforeignlanguage}[2]{{%
\expandafter\ifx\csname l@#1\endcsname\relax
\typeout{** WARNING: IEEEtran.bst: No hyphenation pattern has been}%
\typeout{** loaded for the language `#1'. Using the pattern for}%
\typeout{** the default language instead.}%
\else
\language=\csname l@#1\endcsname
\fi
#2}}
\providecommand{\BIBdecl}{\relax}
\BIBdecl

\bibitem{renzaglia2012multi}
A.~Renzaglia, L.~Doitsidis, A.~Martinelli, and E.~B. Kosmatopoulos,
  ``{Multi-robot three-dimensional coverage of unknown areas},'' \emph{The
  International Journal of Robotics Research}, vol.~31, no.~6, pp. 738--752,
  2012.

\bibitem{bozma2012multirobot}
H.~I. Bozma and M.~E. Kalalioglu, ``{Multirobot coordination in pick-and-place
  tasks on a moving conveyor},'' \emph{Robotics and Computer-Integrated
  Manufacturing}, vol.~28, no.~4, pp. 530--538, 2012.

\bibitem{mathew2015planning}
N.~Mathew, S.~L. Smith, and S.~L. Waslander, ``{Planning paths for package
  delivery in heterogeneous multirobot teams},'' \emph{IEEE Transactions on
  Automation Science and Engineering}, vol.~12, no.~4, pp. 1298--1308, 2015.

\bibitem{korsah2013comprehensive}
G.~A. Korsah, A.~Stentz, and M.~B. Dias, ``{A comprehensive taxonomy for
  multi-robot task allocation},'' \emph{The International Journal of Robotics
  Research}, vol.~32, no.~12, pp. 1495--1512, 2013.

\bibitem{khamis2015multi}
A.~Khamis, A.~Hussein, and A.~Elmogy, ``{Multi-robot task allocation: A review
  of the state-of-the-art},'' \emph{Cooperative robots and sensor networks
  2015}, pp. 31--51, 2015.

\bibitem{Ferguson2021robust}
B.~L. Ferguson and J.~R. Marden, ``{Robust Utility Design in Distributed
  Resource Allocation Problems with Defective Agents},'' \emph{2021 60th IEEE
  Conference on Decision and Control (CDC)}, pp. 1650--1655, dec 2021.

\bibitem{Marden2017}
J.~R. Marden, ``{The Role of Information in Distributed Resource Allocation},''
  \emph{IEEE Transactions on Control of Network Systems}, vol.~4, no.~3, pp.
  654--664, 2017.

\bibitem{agmon2008multi}
N.~Agmon, S.~Kraus, and G.~A. Kaminka, ``{Multi-robot perimeter patrol in
  adversarial settings},'' in \emph{2008 IEEE International Conference on
  Robotics and Automation}.\hskip 1em plus 0.5em minus 0.4em\relax IEEE, 2008,
  pp. 2339--2345.

\bibitem{Kang2013}
M.~S. Kang, S.~B. Lee, and V.~D. Gligor, ``{The Crossfire Attack},'' in
  \emph{2013 IEEE Symposium on Security and Privacy}, 2013, pp. 127--141.

\bibitem{grimsman2020stackelberg}
D.~Grimsman, J.~P. Hespanha, and J.~R. Marden, ``{Stackelberg equilibria for
  two-player network routing games on parallel networks},'' in \emph{2020
  American Control Conference (ACC)}.\hskip 1em plus 0.5em minus 0.4em\relax
  IEEE, 2020, pp. 5364--5369.

\bibitem{brown2006defending}
G.~Brown, M.~Carlyle, J.~Salmer{\'{o}}n, and K.~Wood, ``{Defending critical
  infrastructure},'' \emph{Interfaces}, vol.~36, no.~6, pp. 530--544, 2006.

\bibitem{Abdallah2021}
M.~Abdallah, T.~Cason, S.~Bagchi, and S.~Sundaram, ``{The Effect of Behavioral
  Probability Weighting in a Simultaneous Multi-Target Attacker-Defender
  Game},'' in \emph{2021 European Control Conference (ECC)}, 2021, pp.
  933--938.

\bibitem{kovenock2018optimal}
D.~Kovenock and B.~Roberson, ``{The optimal defense of networks of targets},''
  \emph{Economic Inquiry}, vol.~56, no.~4, pp. 2195--2211, 2018.

\bibitem{Duvallet2010}
F.~Duvallet and A.~Stentz, ``{Imitation learning for task allocation},'' in
  \emph{2010 IEEE/RSJ International Conference on Intelligent Robots and
  Systems}, 2010, pp. 3568--3573.

\bibitem{roberson2006colonel}
B.~Roberson, ``{The colonel blotto game},'' \emph{Economic Theory}, vol.~29,
  no.~1, pp. 1--24, 2006.

\bibitem{tambe2011security}
M.~Tambe, \emph{{Security and game theory: algorithms, deployed systems,
  lessons learned}}.\hskip 1em plus 0.5em minus 0.4em\relax Cambridge
  university press, 2011.

\bibitem{pita2008deployed}
J.~Pita, M.~Jain, J.~Marecki, F.~Ord{\'{o}}{\~{n}}ez, C.~Portway, M.~Tambe,
  C.~Western, P.~Paruchuri, and S.~Kraus, ``{Deployed armor protection: the
  application of a game theoretic model for security at the los angeles
  international airport},'' in \emph{Proceedings of the 7th international joint
  conference on Autonomous agents and multiagent systems: industrial track},
  2008, pp. 125--132.

\bibitem{xu2020stay}
L.~Xu, S.~Gholami, S.~McCarthy, B.~Dilkina, A.~Plumptre, M.~Tambe, R.~Singh,
  M.~Nsubuga, J.~Mabonga, M.~Driciru, and Others, ``{Stay ahead of Poachers:
  Illegal wildlife poaching prediction and patrol planning under uncertainty
  with field test evaluations (Short Version)},'' in \emph{2020 IEEE 36th
  International Conference on Data Engineering (ICDE)}.\hskip 1em plus 0.5em
  minus 0.4em\relax IEEE, 2020, pp. 1898--1901.

\bibitem{adamo2009blotto}
T.~Adamo and A.~Matros, ``{A Blotto game with incomplete information},''
  \emph{Economics Letters}, vol. 105, no.~1, pp. 100--102, 2009.

\bibitem{kovenock2011blotto}
D.~Kovenock and B.~Roberson, ``{A Blotto game with multi-dimensional incomplete
  information},'' \emph{Economics Letters}, vol. 113, no.~3, pp. 273--275,
  2011.

\bibitem{paarporn2021general}
K.~Paarporn, R.~Chandan, M.~Alizadeh, and J.~R. Marden, ``{A General Lotto game
  with asymmetric budget uncertainty},'' \emph{arXiv preprint
  arXiv:2106.12133}, 2021.

\bibitem{Paarporn2019b}
------, ``{Characterizing the interplay between information and strength in
  Blotto games},'' in \emph{2019 IEEE 58th Conference on Decision and Control
  (CDC)}, 2019, pp. 5977--5982.

\bibitem{shishika2021dynamic}
D.~Shishika, Y.~Guan, M.~Dorothy, and V.~Kumar, ``{Dynamic Defender-Attacker
  Blotto Game},'' \emph{arXiv preprint arXiv:2112.09890}, 2021.

\bibitem{shishika2020review}
D.~Shishika and V.~Kumar, ``{A review of multi agent perimeter defense
  games},'' in \emph{International Conference on Decision and Game Theory for
  Security}.\hskip 1em plus 0.5em minus 0.4em\relax Springer, 2020, pp.
  472--485.

\bibitem{chen2014path}
M.~Chen, Z.~Zhou, and C.~J. Tomlin, ``{A path defense approach to the
  multiplayer reach-avoid game},'' in \emph{53rd IEEE conference on decision
  and control}.\hskip 1em plus 0.5em minus 0.4em\relax IEEE, 2014, pp.
  2420--2426.

\bibitem{garcia2019strategies}
E.~Garcia, A.~{Von Moll}, D.~W. Casbeer, and M.~Pachter, ``{Strategies for
  defending a coastline against multiple attackers},'' in \emph{2019 IEEE 58th
  Conference on Decision and Control (CDC)}.\hskip 1em plus 0.5em minus
  0.4em\relax IEEE, 2019, pp. 7319--7324.

\bibitem{ESLee2020}
E.~S. Lee, D.~Shishika, and V.~Kumar, ``Perimeter-defense game between aerial
  defender and ground intruder,'' in \emph{2020 59th IEEE Conference on
  Decision and Control (CDC)}, 2020, pp. 1530--1536.

\bibitem{Macharet2020}
D.~G. Macharet, A.~K. Chen, D.~Shishika, G.~J. Pappas, and V.~Kumar, ``Adaptive
  partitioning for coordinated multi-agent perimeter defense,'' in \emph{2020
  IEEE/RSJ International Conference on Intelligent Robots and Systems (IROS)},
  2020, pp. 7971--7977.

\bibitem{AKChen2021}
A.~K. Chen, D.~G. Macharet, D.~Shishika, G.~J. Pappas, and V.~Kumar, ``Optimal
  multi-robot perimeter defense using flow networks,'' in \emph{Distributed
  Autonomous Robotic Systems}, F.~Matsuno, S.-i. Azuma, and M.~Yamamoto,
  Eds.\hskip 1em plus 0.5em minus 0.4em\relax Cham: Springer International
  Publishing, 2022, pp. 282--293.

\bibitem{Shishika2021p}
D.~Shishika, D.~Maity, and M.~Dorothy, ``Partial information target defense
  game,'' in \emph{2021 IEEE International Conference on Robotics and
  Automation (ICRA)}, 2021, pp. 8111--8117.

\end{thebibliography}

\end{document}